\documentclass[runningheads]{llncs}

\usepackage[utf8]{inputenc}
\usepackage{microtype} 
\usepackage{lmodern}
\usepackage[T1]{fontenc}
\usepackage{csquotes}
\DeclareUnicodeCharacter{2212}{-}

\usepackage{enumitem}
\usepackage{booktabs}

\usepackage{framed}

\usepackage{graphicx}
\usepackage{comment}
\usepackage{cite}

\usepackage{amsmath}
\usepackage{amssymb}
\usepackage{mathtools}
\usepackage{amsfonts}
\usepackage[c2, nocomma, short]{optidef}
\usepackage{nicefrac}

\usepackage[most]{tcolorbox}

\usepackage{color}
\definecolor{linkcolor}{rgb}{0, 0.25, 0.75}

\usepackage{hyperref}
\hypersetup{
  colorlinks=true,
 linkcolor=linkcolor,
  citecolor=linkcolor,
  urlcolor=linkcolor,
 bookmarksnumbered
}
\usepackage[nameinlink, capitalize]{cleveref}

\usepackage{tikz}
\usepackage{pgfplots}
\pgfplotsset{compat=1.17}

\usepackage{graphicx}
\usepackage[labelfont=bf]{caption}
\usepackage{subcaption}

\usepackage{algorithm}
\usepackage[beginLComment=//~,endLComment=,beginComment=//~,endComment=]{algpseudocodex}

\crefname{theorem}{Theorem}{Theorems}
\crefname{corollary}{Corollary}{Corollaries}
\crefname{lemma}{Lemma}{Lemmas}
\crefname{claim}{Claim}{Claims}
\crefname{observation}{Observation}{Observations}
\crefname{proposition}{Proposition}{Propositions}
\crefname{property}{Property}{Properties}
\crefname{proplisti}{Property}{Properties}
\crefname{conjecture}{Conjecture}{Conjectures}
\crefname{remark}{Remark}{Remarks}
\crefname{Invariant}{Invariant}{Invariants}
\crefname{problem}{Problem}{Problems}
\crefname{definition}{Definition}{Definitions}
\crefname{algorithm}{Algorithm}{Algorithms}
\crefname{lp}{}{}
\crefname{figure}{Figure}{Figures}

\spnewtheorem{Invariant}{Invariant}{\bfseries}{\itshape}

\newlist{proplist}{enumerate}{1} 
\setlist[proplist]{label=\arabic*),ref=\theproperty.\arabic*}
\setlist[enumerate,1]{label={\arabic*)}}

\newlist{casesp}{enumerate}{3} 
\setlist[casesp]{align=left, 
                 listparindent=\parindent, 
                 parsep=\parskip, 
                 font=\normalfont\bfseries, 
                 leftmargin=0pt, 
                 labelwidth=0pt, 
                 itemindent=.4em,labelsep=.4em, 
                 partopsep=0pt, 
                 }
\setlist[casesp,1]{label=Case~\arabic*:,ref=\arabic*}
\setlist[casesp,2]{label=Case~\thecasespi.\roman*:,ref=\thecasespi.\roman*}
\setlist[casesp,3]{label=Case~\thecasespii.\alph*:,ref=\thecasespii.\alph*}

\newcommand{\maxwbf}{\textsf{MW$b$-F}}
\newcommand{\maxwbm}{\textsf{MW$b$-M}}
\newcommand{\maxcm}{\textsf{Max Cardinality Matching}}
\newcommand{\maxwm}{\textsf{Max Weight Matching}}
\newcommand{\mwm}{\textsf{MWM}}
\newcommand{\mwbm}{\textsf{MW$b$-M}}

\newcommand{\assign}{\textsf{Assignment}}
\newcommand{\transport}{\textsf{Transportation}}

\newcommand{\bfactor}{$b$-factor}
\newcommand{\bmatching}{$b$-matching}
\newcommand{\bfactors}{$b$-factors}
\newcommand{\bmatchings}{$b$-matchings}
\newcommand{\bfunc}{$b \colon V \to \mathbb{Z}_+$}
\newcommand{\bv}[1][v]{\ensuremath{b(#1)}}

\newcommand{\lrp}[1]{\left( #1 \right)}

\newcommand{\lrc}[1]{\left\{ #1 \right\}}
\newcommand{\lrv}[1]{\left\langle #1 \right\rangle}
\newcommand{\abs}[1]{\left\lvert #1 \right\rvert}

\newcommand{\floor}[1]{{\lfloor {#1} \rfloor}}

\newcommand{\R}{\ensuremath{\mathbb{R}}}

\DeclareMathOperator*{\argmin}{arg\,min}

\newcommand{\bigO}[2][]{\ensuremath{O_{#1}(#2)}}
\newcommand{\bigOlog}[2][]{\ensuremath{\tilde{O_{#1}}(#2)}}
\newcommand{\bigTheta}[2][]{\ensuremath{\Theta_{#1}(#2)}}

\newcommand{\f}[2][f]{\ensuremath{#1 \! \lrp{#2}}}

\newcommand{\degree}[2][]{\f[\mathrm{deg}_{#1}]{#2}}

\def\sse{\subseteq}
\def\sm{\setminus}

\newcommand{\AlgIn}{\Statex \textbf{Input:} }
\newcommand{\AlgOut}{\Statex \textbf{Output:} }

\newcommand{\crefdefpart}[2]{%
    \hyperref[#2]{\namecref{#1}~\labelcref*{#1}.\ref*{#2}}%
}

\makeatletter
\newcommand{\leqnomode}{\tagsleft@true}
\newcommand{\reqnomode}{\tagsleft@false}
\makeatother

\makeatletter
\newcommand{\specificthanks}[1]{\@fnsymbol{#1}}
\makeatother

\makeatletter
\AtBeginDocument{%
  \def\doi#1{\url{https://doi.org/#1}}}
\makeatother

\begin{document}
\title{Approximate Bipartite $b$-Matching using Multiplicative Auction\thanks{The manuscript is accepted as a refereed paper in the 2024 INFORMS Optimization Society conference.}}
%
%

\author{Bhargav Samineni\inst{1}
\and
S M Ferdous\inst{2}
\and Mahantesh Halappanavar\inst{2}
\and
Bala Krishnamoorthy\inst{3}
}

\authorrunning{B. Samineni et al.}
%
\institute{
The University of Texas at Austin, \email{sbharg@utexas.edu}
\and
Pacific Northwest National Laboratory, \email{\{sm.ferdous,hala\}@pnnl.gov} 
\and
Washington State University, \email{kbala@wsu.edu}
}

\maketitle

\begin{abstract}
Given a bipartite graph $G(V= (A \cup B),E)$ with $n$ vertices and $m$ edges and a function $b \colon V \to \mathbb{Z}_+$, a \emph{$b$-matching} is a subset of edges such that every vertex $v \in V$ is incident to at most $b(v)$ edges in the subset. When we are also given edge weights, the \textsf{Max Weight $b$-Matching} problem is to find a $b$-matching of maximum weight, which is a fundamental combinatorial optimization problem with many applications. Extending on the recent work of Zheng and Henzinger (IPCO, 2023) on standard bipartite matching problems, we develop a simple \emph{auction} algorithm to approximately solve \textsf{Max Weight $b$-Matching}. Specifically, we present a multiplicative auction algorithm that gives a $(1 - \varepsilon)$-approximation in $O(m \varepsilon^{-1} \log \varepsilon^{-1} \log \beta)$ worst case time, where $\beta$ the maximum $b$-value. Although this is a $\log \beta$ factor greater than the current best approximation algorithm by Huang and Pettie (Algorithmica, 2022), it is considerably simpler to present, analyze, and implement. 

\end{abstract}

\section{Introduction}
\emph{Matching}, also known as the {\em assignment} problem, in a bipartite graph is one of the most fundamental discrete optimization problems, which is rich in theory, algorithms, and applications. 
A weighted matching (also known as the linear sum assignment problem) in a bipartite graph aims to find a pairing of vertices between the partitions with the maximum sum of edge weights in the matching. There are many classic applications of maximum weighted bipartite matchings (\mwm{}) in resource allocation and assignments~\cite{daglib/0022248,lovasz2009matching}, and we also observe many emerging applications~\cite{olschowka1996new,siammax/DuffK01,azad2020distributed,mehta2013online,deng2017silkmoth,zeakis2022tokenjoin,abeywickrama2021optimizing,simonovsky2018graphvae}.

Classic \emph{exact} algorithms for \mwm{} include the celebrated primal-dual Hungarian algorithm~\cite{kuhn1955hungarian,munkres1957algorithms}, which is expensive and has little to no parallelism. Alternative \emph{approximate} approaches include a class of algorithms called \emph{auction} algorithms~\cite{bertsekas1979distributed,bertsekas_new_1981,demange1986multi}, which treat the matching process as a welfare-maximizing allocation of \emph{objects} to \emph{bidders}, where objects and bidders are vertices in the bipartitions of the graph. Auction algorithms assign prices to the objects and let \emph{eligible} bidders bid on objects with the maximum utility. 
The bidder outbids any previous bids for the object making the previous bidder eligible again. By formulating the bid values in a certain manner, the auction eventually terminates at an equilibrium where all the bidders are \emph{approximately} happy with the object they win. Auction based approaches are easier to present, analyze, and implement as they only involve 
performing a series of simple local updates with good empirical performance~\cite{alfaro2022assignment, ramshaw2012minimum, sathe2012auction}. Often, the runtime complexity of the classic auction algorithms is pseudo-polynomial as they depend on the maximum weight of the graph. There is a renewed interest in auction algorithms focusing on improving the runtime and developing algorithms in \emph{scalable computational models} such as distributed, streaming, and parallel models~\cite{assadi2021auction,ke2023scalable,zheng2023multiplicative}. In particular, a recent result of Zheng and Henzinger~\cite{zheng2023multiplicative} shows a multiplicative auction algorithm achieving a $(1-\varepsilon)$-approximate matching in $\bigO{m \varepsilon^{-1} \log \varepsilon^{-1}}$ time, which matches the more complicated state-of-the-art algorithm of Duan and Pettie~\cite{duan2014linear}.

In this paper, our focus is on the bipartite \bmatching{} problem, which generalizes a matching by allowing each vertex $v \in V$ to be incident to at most $b(v)$ edges in the matching. Weighted bipartite $b$-matchings (\maxwbm{}) are particularly suited for recommendation and assignment applications where multiple choices are preferred. This naturally occurs in modern applications like movie recommendations, route suggestions, and ad allocations. Consequently, \maxwbm{} has been used in  protein
structure alignment~\cite{krissinel2004secondary}, computer
vision~\cite{belongie2002shape}, estimating text similarity~\cite{pang2016text}, reviewers assignment in peer-review systems~\cite{charlin2013toronto,liu2014robust,tang2010expertise}, and diverse assignment~\cite{ahmed2017diverse,ahmadi2020diverse}. 
Although \maxwbm{} has been extensively explored in different algorithmic paradigms, little is known in terms of auction algorithms. 

To extend the auction paradigm to \bmatchings{}, we must deal with
the issue of objects being matched to multiple bidders. Indeed, this breaks the analogy of an auction, as it makes little sense for multiple people to win a single object. 
To deal with this, we make a simple modification: 
instead of allowing up to $b(j)$ bidders to bid on and win a single object $j$, 
we create a set $M(j)$ of $b(j)$ identical \emph{copies} of $j$ for bidders to bid on. 
Additionally, instead of an object $j$ explicitly maintaining a price, each copy of $j$ maintains its own price.
The auction now works by matching bidders to copies of an object, potentially outbidding the current match, and then updating the prices of the object copies. 

Using this modification and adapting the multiplicative auction algorithm of Zheng and Henzinger~\cite{zheng2023multiplicative}, we design an $(1-\varepsilon)$-approximate auction algorithm for \maxwbm{}. In particular, we interpret the auction process for \bmatching{} in a primal-dual linear programming framework. 
In \cref{sec:primal_dual}, we describe \emph{approximate} complementary slackness conditions and prove that any primal-dual variable pairs that obey them have a desired approximation guarantee. In \cref{sec:mult_auction}, we present a $(1-\varepsilon)$-approximation algorithm based on this analysis and the modifications discussed earlier by extending the recent multiplicative auction algorithm of Zheng and Henzinger \cite{zheng2023multiplicative} to \maxwbm{}. The worst case running time is shown as $\bigO{m \varepsilon^{-1} \log \varepsilon^{-1} \log \beta}$, which is a $\log \beta$ factor greater than the running time of the state-of-the-art approximation algorithm by Huang and Pettie \cite{huang2022approximate} when restricted to bipartite graphs. 
While our algorithm has a runtime dependence on $\beta$, it is reasonable to assume that $\beta = \bigO{1}$ for many real-world applications.
Also, the algorithm of Huang and Pettie \cite{huang2022approximate} is relatively complicated to analyze and implement even for bipartite graphs, as it is based on the scaling framework \cite{duan2014linear,gabow1985scaling}. In contrast, the multiplicative auction algorithm we present is considerably simpler in both regards.

\section{Preliminaries} \label{sec:prelims}
Let $G = \lrp{V = \lrp{A \cup B}, E, w}$ be a simple undirected bipartite graph with bipartitions $A$ and $B$, 
$n \coloneqq \abs{V}$ vertices, $m \coloneqq \abs{E}$ edges,
and weights $w \colon E \to \R_{\geq 0}$. 
For a vertex $v \in V$, denote by $\degree{v}$ and $N(v)$ the number of edges it is incident to and its set of neighbors in $G$, respectively.  
For a subset of edges $H \sse E$, let $\degree[H]{v}$ denote the number of edges in $H$ that $v$ is incident to. 
We define $\Delta \coloneqq \max_{v \in V} \degree{v}$. 
Given a function \bfunc{}, a \bmatching{} (also known as $f$-matching or degree-constrained subgraph)
is a subset of edges $F \sse E$ such that 
$\degree[F]{v} \leq \bv$ for all $v \in V$, where we can assume without loss of generality $1 \leq \bv \leq \degree{v}$.
We denote $\beta \coloneqq \max_{v \in V} \bv$. 
 A vertex $v$ is \emph{saturated} by $F$ if $\degree[F]{v} = \bv$, 
and it is \emph{unsaturated} 
by $F$ if $\degree[F]{v} < \bv$. 
Additionally, we let $F(v)$ denote the set of vertices $v$ is matched to under $F$.  
For a real-valued function $f$ defined on the elements of a set $Y$,
we use the standard summing notation $f(Y) \coloneqq \sum_{y \in Y} f(y)$. 
Without loss of generality, we assume $\bv[A] \leq \bv[B]$; thus the size of any \bmatching{} is at most $\bv[A] \leq \frac{\bv[V]}{2}$. 
The \textsf{Max Weight $b$-Matching} (\maxwbm{}) problem is to find a $b$-matching $F$ that maximizes $w(F)$ given a weighted bipartite graph $G = (V = \lrp{A \cup B}, E, w)$ and function \bfunc{} as input. 

\section{Related Work} \label{sec:prior_work}

The auction approach for \maxwm{} (\mwm{}) can be
attributed to Demange et al.~\cite{demange1986multi} and Bertsekas \cite{bertsekas1979distributed,bertsekas_new_1981}, 
who also extended it to the \assign{}, \transport{}, and general network flow problems \cite{bertsekas1988auction, bertsekas1992auction}. 
Recently, Assadi et al.~\cite{assadi2021auction} gave an auction algorithm for $(1-\varepsilon)$-approximate
\maxcm{} that yields algorithms in the semi-streaming~\cite{feigenbaum2005graph} and MPC~\cite{karloff2010model} models of computation. 
Liu et al.~\cite{ke2023scalable} extended on this work to develop auction algorithms for  $(1 - \varepsilon)$-approximate \mwm{} and \textsf{Max Cardinality $b$-Matching}
that work in various scalable models of computation. We note that the algorithm of Liu et al.~\cite{ke2023scalable} is the only auction approach for \bmatching{} that we are aware of, and it works only for the unweighted version of the problem. Additionally, Zheng and Henzinger \cite{zheng2023multiplicative} developed multiplicative auction algorithms 
to give a $\bigO{m \varepsilon^{-1} \log \varepsilon^{-1}}$ 
time auction algorithm for $(1 - \varepsilon)$-approximate \mwm{}.

Several algorithms for special cases of \mwbm{} exist.
For integral edge weights, Gabow and Tarjan \cite{gabow1989faster} 
developed exact scaling algorithms by reducing to finding a perfect \bmatching{} 
on a bipartite multigraph, while
Huang and Kavitha \cite{huang2017new} give an exact $\bigO{\bv[V]^{1/2} mW}$ time algorithm by decomposing into $W$ unweighted
subproblems.
In general graphs, finding a max weight \bmatching{} is also well studied.
Gabow \cite{gabow1983efficient} gave an exact $\bigO{\bv[V] \min \{m \log n, n^2\}}$ time algorithm.
Bayati et al.~\cite{bayati2011belief} proposed an exact algorithm based on belief propagation. Huang and Pettie \cite{huang2022approximate} presented a $\lrp{1-\varepsilon}$-approximate scaling algorithm with running time $\bigO{m \alpha(m, n) \varepsilon^{-1} \log \varepsilon^{-1}}$, 
where $\alpha(m, n)$ is the inverse Ackermann function. For bipartite graphs, the running time reduces to $\bigO{m \varepsilon^{-1} \log \varepsilon^{-1}}$. There are also several $\frac{1}{2}$ and $(\frac{2}{3}-\varepsilon)$-approximation algorithms designed for \bmatching{}~\cite{mestre2006greedy,khan2016efficient,pothen2019approximation} in the sequential and parallel models. 
Finding a max weight \bmatching{} can also be reduced to a standard matching problem \cite{huang2022approximate, shiloach1981another,ferdous2021algorithms}.
However, the reduction is not approximation preserving, and may drastically increase the size of the graph. 

\section{Primal-Dual Analysis for $b$-Matchings} \label{sec:primal_dual}

Given a graph $G = \lrp{V = \lrp{A \cup  B}, E, w}$, we refer to vertices in $A$ as \emph{bidders} and vertices in $B$ as \emph{objects}. 
For each edge $(i, j) \in E$, define an indicator variable $x(i,j) \in \lrc{0, 1}$.
The LP-Relaxation of \maxwbm{}
and its dual are then given by

\vspace{-2ex} 
\leqnomode
\noindent\begin{minipage}[t]{.45\linewidth}
    \footnotesize
    \begin{maxi*}
        {}{\sum_{(i, j) \in E} w(i,j) x(i,j)}{}{}{} 
        \addConstraint{\sum_{j \in N(i)} x(i,j)}{\leq \bv[i]}{\; \forall i \in A}
        \addConstraint{\sum_{i \in N(j)} x(i,j)}{\leq \bv[j]}{\; \forall j \in B} 
        \addConstraint{0 \leq x(i,j)}{\leq 1}{\; \forall \lrp{i,j} \in E}.
    \end{maxi*}%
\end{minipage}%
\hspace{-1em}
\begin{minipage}[t]{.59\linewidth}
    \footnotesize
    \begin{mini*}
        {}{\sum_{i \in A} \bv[i] \pi(i) + \sum_{j \in B} \bv[j] p(j) + \sum_{(i,j) \in E} z(i, j)}{}{}{} 
        \addConstraint{\pi(i) + p(j) + z(i, j)}{\geq w(i,j)}{\; \forall (i, j) \in E}
        \addConstraint{\pi(i)}{ \geq 0}{\; \forall i \in A}
        \addConstraint{p(j)}{ \geq 0}{\; \forall j\in B}
        \addConstraint{z(i, j)}{ \geq 0}{\; \forall \lrp{i,j} \in E}.
    \end{mini*}
\end{minipage}
\reqnomode

\vspace{1em}
The dual variables $z$ are defined for every edge, while the dual variables $\pi$ and $p$ are defined for vertices in $A$ and $B$, respectively. 
We refer to the dual $\pi(i)$ for a bidder $i$ as its \emph{profit} and the dual $p(j)$ for an object $j$ as its \emph{price}. 

\begin{property}[Complementary Slackness]
    \label{prop:comp_slack} 
    Let $x$ and $(\pi,p, z)$ be feasible primal and dual solutions,
    and let $F$ be the \bmatching{} induced by $x$. 
    Then these are \emph{optimal} solutions if and only if the following conditions hold:
    \begin{proplist}
        \item \label{comp_slack_prop1} $x(i,j) > 0 \implies \pi(i) + p(j) + z(i, j) = w(i,j)$ for all $(i, j) \in  E$.
        \item \label{comp_slack_prop2} $x(i,j) < 1 \implies z(i, j) = 0$ for all $(i, j) \in  E$.
        \item \label{comp_slack_prop3} $\sum_{j \in N(i)} x(i, j) < b(i) \implies \pi(i) = 0$ for all $i \in A$.
        \item \label{comp_slack_prop4} $\sum_{i \in N(j)} x(i, j) < b(j) \implies p(j) = 0$ for all $j \in B$.
    \end{proplist}
    The first two conditions can be restated as $\pi(i) + p(j) \leq w(i,j)$ if $(i, j) \in F$ and 
    $\pi(i) + p(j) \geq w(i,j)$ if $(i, j) \notin F$, respectively, while the last two conditions indicate 
    any unsaturated vertices have an optimal dual value of zero. 
\end{property} 
By \cref{comp_slack_prop1,comp_slack_prop2}, maintaining the edge duals $z$
is redundant as their minimizing value can be given explicitly by 
$z(i, j) = \max \lrc{w(i,j) - \pi(i) - p(j), 0}$. 
We can also partially restate the conditions for all $(i, j) \in F$ as
\begin{equation}
    w(i,j) - p(j) \geq \pi(i) \geq \max \lrc{ \max_{k \in N(i) \sm F(i)} \lrc{w(i,k) - p(k)}, 0 },
    \label{eq:opt-comp-slack}
\end{equation}
where $F(i)$ is the set of objects $i$ is matched to under $F$. The upper and lower bound on $\pi(i)$
follows from the slackness condition when $(i,j) \in F$ and when $(i,j) \notin F$, respectively, along with 
the non-negativity constraints of the dual variables. 
Additionally, owing to \cref{comp_slack_prop3}, if a bidder $i$ is unsaturated in an optimal solution, then 
$\pi(i) = 0$ which implies $\max_{k \in N(i) \sm F(i)} \lrc{w(i,k) - p(k)} \leq 0$. 
Hence, if we are given some arbitrary prices $p$, a feasible value of $\pi$ respecting 
\cref{prop:comp_slack} is implicitly given by the lower bound of \cref{eq:opt-comp-slack}.

This motivates the naive auction based approach, where both the primal and dual problems are 
simultaneously solved, but only the primal and price variables are maintained explicitly.  
Consider a set of matched edges $F$ and prices $p$. \cref{eq:opt-comp-slack} 
motivates the construction of $F$ and $p$ such that every edge $(i, j) \in F$ satisfies
\begin{equation*}    
    w(i,j) - p(j) \geq \max \lrc{ \max_{k \in N(i) \sm F(i)} \lrc{w(i,k) - p(k)}, 0 } \eqqcolon \pi(i). 
\end{equation*}
We call such an edge \emph{happy}. For a bidder $i$, if all edges $(i, j) \in F$ are happy and either $i$ is saturated by $F$ and $\pi(i) \geq 0$
or $i$ is unsaturated by $F$ and $\pi(i) = 0$, then we also call $i$ happy. 
If all bidders are happy under $F$, 
then $F$ is also happy. If $F$ is happy, a feasible \bmatching{}, and together with prices $p$ satisfies \cref{comp_slack_prop4} (i.e.~unsaturated objects have a dual value of zero),
then it must be optimal since it satisfies every condition of \cref{prop:comp_slack}.

\subsection{\texorpdfstring{$\varepsilon$}{\unichar{"03B5}}-Happiness}

We can relax the notion of happiness to an \emph{approximate} sense by maintaining some multiplicative slack $(1-\varepsilon)$ for some $\varepsilon > 0$.  
\begin{definition}[$\varepsilon$-Happiness]
    \label{def:epsilon_happy_edge_mult}
    For $F \sse E$ and prices $p$, 
    an edge $(i, j) \in F$ is $\varepsilon$-happy if   
    \begin{equation*}
        w(i,j) - p(j) \geq \max \lrc{ \max_{k \in N(i) \sm F(i)} \lrc{(1-\varepsilon)w(i,k) - p(k)}, 0 } \eqqcolon \pi(i).
    \end{equation*}
    A bidder $i$ is $\varepsilon$-happy if all edges $(i, j) \in F$ are $\varepsilon$-happy and either 
    $i$ is saturated by $F$ and $\pi(i) \geq 0$ \textbf{or} $i$ is unsaturated by $F$ and $\pi(i) = 0$. 
\end{definition}

Note that $F$ being $\varepsilon$-happy satisfies \cref{comp_slack_prop3}. 
If $F$ is $\varepsilon$-happy, a feasible \bmatching{}, and together with prices $p$ satisfies \cref{comp_slack_prop4}, 
then its weight falls within a $(1-\varepsilon)$ factor of the max weight \bmatching{}.  
To show this, we use a lemma from Huang and Pettie \cite{huang2022approximate} which follows from simple primal-dual arguments. 
We note that the original lemma is for \bmatchings{} in general graphs, but we specialize it for bipartite graphs and our notation here. 

\begin{lemma}[{\cite[Lemma 5]{huang2022approximate}}]
    Let $F$ be a \bmatching{} and $\pi$, $p$ the dual variables for vertices in $A$ and $B$, respectively. 
    If all unsaturated vertices have dual values of zero, each unmatched edge $(i, j) \notin F$ satisfies 
    $\pi(i) + p(j) \geq (1-\delta_1)w(i, j)$, and each matched edge $(i, j) \in F$ satisfies $\pi(i) + p(j) \leq (1+\delta_2)w(i, j)$, then 
    $w(F) \geq (1-\delta_1)(1 + \delta_2)^{-1}w(F^{\ast})$ where $F^{\ast}$ is a max weight \bmatching{}.  
    \label{lem:hp_mult_error}
\end{lemma}

\begin{lemma}
    Let $F, p$ be some feasible \bmatching{} and prices, respectively, such that $F$ is $\varepsilon$-happy 
    and $F, p$ satisfy \cref{comp_slack_prop4}. Then $w(F) \geq (1-\varepsilon)w(F^{\ast})$.
    \label{lem:mult_loss_bmatching}
\end{lemma}
\begin{proof}
    For $(i, j) \in F$, we have that
    $
        \pi(i) + p(j) \leq w(i, j) - p(j) + p(j) = w(i, j)
    $
    since $(i, j)$ is $\varepsilon$-happy. For $(i, j) \notin F$, we have that  
    \begin{align*}
        \pi(i) + p(j) & = \max \lrc{\max_{k \in N(i) \sm F(i)} \lrc{(1-\varepsilon)w(i,k) - p(k)},0} + p(j)
        \\ &\geq (1-\varepsilon)w(i, j) - p(j) + p(j) = (1-\varepsilon)w(i, j)
    \end{align*}
    since $i$ is $\varepsilon$-happy and $j \in N(i) \sm F(i)$. 
    By \cref{comp_slack_prop4} and $F$ being $\varepsilon$-happy, the dual values 
    of unsaturated vertices are zero. 
    The lemma follows by applying \cref{lem:hp_mult_error} with $\delta_1 = \varepsilon$, $\delta_2 = 0$. 
    \qed
\end{proof}
\section{A Multiplicative Auction Algorithm} \label{sec:mult_auction}

The notion of $\varepsilon$-happiness gives us a framework 
to build an algorithm with approximation guarantees 
as long as we can maintain that matched edges are $\varepsilon$-happy. 
However, 
updating the prices of objects is non-trivial since an object may be matched to multiple bidders, and updating the prices carelessly may break $\varepsilon$-happiness for a matched edge. 
To resolve this, we introduce the idea of object copies. Specifically, we associate each object $j \in B$ with a 
set of object copies $M(j) = \{ c_1, \ldots, c_{b(j)} \}$, where each object copy $c \in M(j)$ maintains its own price
$p_j(c)$. When a bidder wants to match with an object $j$, it must choose a 
specific object copy $c \in M(j)$ to bid on and be assigned to. Note that an object copy can be assigned to at most one bidder. 

\begin{figure}[t]
    \centering
    \begin{subfigure}[b]{0.35\textwidth}
        \centering
        \includegraphics[scale=0.7250]{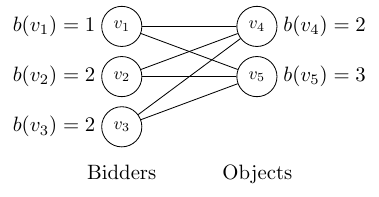}
        \caption{The original graph.}
    \end{subfigure}%
    \begin{subfigure}[b]{0.65\textwidth}
        \centering
        \includegraphics[scale=0.7250]{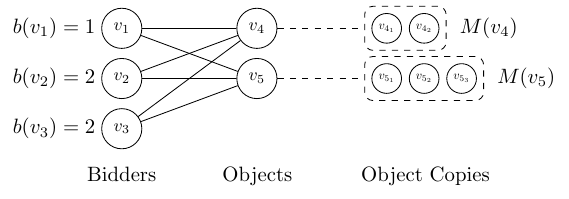}
        \caption{Adding copies of objects.}
    \end{subfigure}
    \caption{The set of copies associated with each object are given by dashed boxes. During the auction process,
    bidders bid on and increase the price of a specific copy of an object in order to match to it.}
    \label{fig:example}
\end{figure}

We visualize the role of multiple copies of an object in our algorithm in \cref{fig:example}.  
In this example, an eligible bidder, say $v_2$, would search through the copies of $v_4$ and $v_5$, 
and then bid on the ones that offer it the best utilities in each.
The other eligible bidders would work similarly, potentially outbidding and unmatching someone else on the same object copy.  
This process continues until an equilibrium is reached. 
This modification allows us to maintain the overall auction process described for matching problems, 
with the only changes being that bidders may now bid on multiple things and have to increase the price of a
specific copy of $j$ instead of $j$ directly to match to it. 
The addition of these object copies motivates a slight change to $\varepsilon$-happiness, which we call \emph{strong $\varepsilon$-happiness}. 

\begin{definition}[Strong $\varepsilon$-Happiness]
    For $F \sse E$ and prices for each object, an edge $(i, j) \in F$, where $i$ is assigned to 
    $c \in M(j)$, is strongly $\varepsilon$-happy if, 
    \begin{equation*}
        w(i,j) - p_j(c) \geq \max \lrc{ \max_{k \in N(i) \sm F(i),\; l \in M(k)} \lrc{(1-\varepsilon)w(i,k) - p_k(l)}, 0 } \eqqcolon \pi(i).
    \end{equation*}
\end{definition}

We can also show that strong $\varepsilon$-happiness implies $\varepsilon$-happiness if we set the price of an 
object $j$ as the minimum of the prices of its object copies. 

\begin{proposition}
    Let $F$ be strongly $\varepsilon$-happy. If we set $p(j) = \min_{c \in M(j)} \{p_j(c)\}$ as the price for each object $j \in B$, then 
    $F$ is $\varepsilon$-happy. 
    \label{prop:strong_e_hap_to_e_hap_mult}
\end{proposition}
\begin{proof}
    Consider an edge $\lrp{i, j} \in F$, where $i$ is assigned to $c \in M(j)$.   
    Then, 
    \begin{align*}
        w(i, j) - p(j) \geq w(i, j) - p_j(c) &\geq \max \lrc{ \max_{k \in N(i) \sm F(i),\; l \in M(k)} \lrc{(1-\varepsilon)w(i,k) - p_k(l)}, 0 } \\
        &= \max \lrc{ \max_{k \in N(i) \sm F(i)} \lrc{(1-\varepsilon)w(i,k) - p(k)}, 0 }
    \end{align*}
    where the last equality follows from the fact that
    \begin{equation*}
    (1-\varepsilon)w(i, k) - p(k) = (1-\varepsilon)w(i, k) - \min_{l \in M(k)} {p_k(l)} = \max_{l \in M(k)} {(1-\varepsilon)w(i, k) - p_k(l),}
    \end{equation*}
    for all objects $k \in N(i)$. 
\end{proof}

We now describe our algorithm that adapts the multiplicative auction algorithm of
Zheng and Henzinger \cite{zheng2023multiplicative} to \bmatchings{}. 

\subsection{Weight Preprocessing} \label{subsec:weight}
We first pre-process the edge weights as described in \cite{zheng2023multiplicative}. If we are given an approximate parameter $\varepsilon^\prime \in (0, 1)$, we can afford to round and scale the edge weights since we are 
seeking a $(1-\varepsilon^\prime)$-approximate solution.
We remove any edge of weight less than $\frac{\varepsilon^\prime}{\bv[V]} W$  since including 
$\frac{1}{2} \bv[V]$ such edges in a \bmatching{}\footnote{Note that the size of any \bmatching{} is upper bounded by $\bv[A] \leq \frac{1}{2}\bv[V]$} 
would not even add $\frac{\varepsilon^\prime}{2} W \leq \frac{\varepsilon^\prime}{2} w(F^{\ast})$ to its weight.
Hence, we can scale the original weights by $\frac{1}{\varepsilon^\prime W} \bv[V]$ so that the maximum edge 
weight is $\frac{\bv[V]}{\varepsilon^\prime}$ and the minimum edge weight is at least $1$. We slightly abuse notation and denote these 
pre-processed edge weights with the function $w$. 
Next, we round down the weights to the nearest integer power of $(1+\varepsilon)$, where $\varepsilon = \frac{1}{2} \varepsilon^\prime$. 
To do this, we define $\textsc{iLog}(x) = \floor{\log_{1+\varepsilon}(x)}$ and the new weights as
$\tilde{w}(i, j) = (1 + \varepsilon)^{\textsc{iLog}(w(i, j))}$ for each edge $(i, j) \in E$.
We remark that finding a $(1-\varepsilon)$-approximate solution with respect to $\tilde{w}$ gives 
a $(1-\varepsilon^\prime)$-approximate solution with respect to $w$.

\begin{proposition}
    \label{prop:scaling}
    Let $w$ and $\tilde{w}$ be the weights before and after rounding. To find a 
    $(1-\varepsilon^\prime)$-approximate solution with respect to $w$, it suffices to find 
    a $(1-\varepsilon)$-approximate solution with respect to $\tilde{w}$. 
\end{proposition}
\begin{proof}
    Note that $(1+\varepsilon)^{-1}w(i, j) < \tilde{w}(i, j) \leq w(i, j)$ by definition of the rounding. 
    If $F^{\ast}$ and $\tilde{F}^{\ast}$ are optimal solutions with respect to $w$ and $\tilde{w}$ and 
    $F$ a $(1-\varepsilon)$-approximate solution with respect to $\tilde{w}$, then 
    \begin{align*}
        w(F) \geq \tilde{w}(F) \geq (1-\varepsilon) \tilde{w} (\tilde{F}^{\ast}) 
        \geq (1-\varepsilon) \tilde{w} (F^{\ast})
        > \frac{1-\varepsilon}{1+\varepsilon} w(F^{\ast})
        \geq (1 - 2\varepsilon) w(F^{\ast}). 
    \end{align*}
    Since $\varepsilon = \nicefrac{\varepsilon^{\prime}}{2}$, we get that
    $w(F) \geq (1 - \varepsilon^\prime)w(F^{\ast})$. 
    \qed
\end{proof}

We also take note of two important integers when rounding, 
$s_{\max} = \textsc{iLog}(W) = \textsc{iLog}(\bv[V]/2\varepsilon) = \bigO{\varepsilon^{-1} \log(\bv[V]\varepsilon^{-1})}$ and 
$s_{\min}$, the smallest integer such that $(1+\varepsilon)^{-s_{\min}} \leq \varepsilon$. A simple analysis shown in \cite{zheng2023multiplicative}
gives that $s_{\min} = \bigTheta{\varepsilon^{-1} \log \varepsilon^{-1}}$. 

\subsection{The Algorithm}
We present the pseudocode of our algorithm in~\cref{alg:mult_auction_bmatching}. For each object $j \in B$, we initialize its set of copies and set their prices to zero. 
For each bidder $i \in A$, we build a queue $Q_i$ that contains pairs of the form $(r, (i, j))$
for each $j \in N(i)$ and each integer \emph{index} $r_{ij} - s_{\min} \leq r \leq r_{ij} = \textsc{iLog}(w(i, j))$, 
where the pairs in $Q_i$ are ordered in non-increasing order based on the index $r$ starting from the top of the queue. 
To efficiently build the queues, we sort the pairs associated with all the edges using a global bucket sort (lines \ref{line:gbs_start} to \ref{line:gbs_end}).
With this, we can populate $Q_i$ for each $i \in A$ by going through the indices in decreasing order and inserting any 
relevant pairs (lines \ref{line:qstart} to \ref{line:qend}). We also maintain for each $i \in A$ an integer $r_i$ corresponding to the most recent popped index from $Q_i$. 

We maintain a list $I$ of bidders that are not strongly $\varepsilon$-happy. While this 
list is not empty, we remove a bidder $i$ from it and call the function $\textsc{AssignAndBid}(i)$. This function pops pairs from $Q_i$ until it saturates $i$ or empties $Q_i$ and along the way matches to certain objects. 
More accurately, suppose a pair $(r, (i, j))$ is popped. If $(i, j) \notin F$, then $i$ will attempt to match with the cheapest object copy $c \in M(j)$ given that $\tilde{w}(i, j) - p_j(c)$ is above a certain threshold. This threshold  guarantees that 
matching to $j$ is a \textquote{safe} choice, and we do not have to worry about a better choice coming up later in the queue. 
If $i$ is matched to $c$, we add the tuple $\lrv{j, c}$ to a temporary list $T$ and update the \bmatching{} $F$ to indicate this. 
If $c$ was previously assigned to some other bidder $y$, we remove the relevant edge from $F$ and add $y$ back to $I$ if $Q_y \neq \emptyset$ as it may not be strongly $\varepsilon$-happy. 
Otherwise, if $(i, j) \in F$, where $i$ is assigned to $c \in M(j)$, we add the tuple $\lrv{j, c}$ to $T$ if the value $\tilde{w}(i, j) - p_j(c)$ is also above a certain threshold. 
Once $i$ becomes saturated or $Q_i$ is empty, we calculate bids for all the tuples in $T$ based on the current $r_i$ value and  
perform a price update on all the chosen object copies. Finally, we can remove $i$ from $I$ as we can show it is strongly $\varepsilon$-happy.

\begin{algorithm*}[!t] 
    \caption{Multiplicative Auction} \label{alg:mult_auction_bmatching}

    \begin{minipage}[t]{0.46\linewidth}
    \begin{algorithmic}[1]
        \AlgIn $G = \lrp{V = \lrp{ A \cup B}, E}$, weights $\tilde{w} \colon E \to \R_{\geq 0}$, vertex capacities \bfunc{}, 
        and $\varepsilon \in (0, \frac{1}{2})$
        \AlgOut A set of edges $F$ such that $F$ is a strongly $\varepsilon$-happy \bmatching{}

        \State $I \gets A$, $F \gets \emptyset$
        \State $r_i \gets 0, Q_i \gets \emptyset$ for all $i \in A$ 
        \State $L_r \gets \emptyset$ for all $r$ from $s_{\max}$ to $-s_{\min}$ 
        \For{$j \in B$}
            \State $M(j) \gets \lrc{c_1, \ldots, c_{b(j)}}$  
            \State $p_j(c) \gets 0$ for all $c \in M(j)$
        \EndFor
        \LComment{\textbf{Initialization Phase}}
        \For{$(i, j) \in E$} \label{line:gbs_start}
            \State $r_{ij} \gets \textsc{iLog}(\tilde{w}(i, j))$
            \State $r_i \gets \max \{r_i, r_{ij}\}$ 
            \For{$r$ from $r_{ij}$ to $r_{ij} - s_{\min}$}
                \State $L_r \gets L_r \cup (r, (i, j))$
            \EndFor
        \EndFor \label{line:gbs_end}
        \For{$r$ from $s_{\max}$ to $-s_{\min}$} \label{line:qstart}
            \For{$(r, (i, j)) \in L_r$}
                \State $Q_i$.push($(r, (i, j))$)
            \EndFor
        \EndFor \label{line:qend}
        \LComment{\textbf{Auction Phase}}
        \While{$I \neq \emptyset$}
            \State Choose some $i \in I$
            \State $\Call{AssignAndBid}{i}$
        \EndWhile
        \State \Return $F$
        \end{algorithmic}
        \end{minipage}
        \hfill
        \begin{minipage}[t]{0.5\linewidth}
        \begin{algorithmic}[1]
            \setcounter{ALG@line}{20}
        \Procedure{AssignAndBid}{$i$} 
            \State $T \gets \emptyset$ \Comment{Temporary List}
            \While{$Q_i \neq \emptyset$ and $\abs{F(i)} < \bv[i]$}
                \State $\lrp{r, (i, j)} \gets Q_i$.pop(), $r_i \gets r$ 
                \If{$j \in F(i)$}
                    \State Let $c \in M(j)$ be the object copy $i$ is assigned to 
                    \If{$\tilde{w}(i, j) - p_j(c) \geq (1 + \varepsilon)^r$}
                        \State $T \gets T \cup \lrc{\lrv{j, c}}$
                    \EndIf
                \Else{}
                    \State $c \gets \argmin_{c^\prime \in M(j)} {p_j(c^\prime)}$ \label{line:obj_copy}
                    \If{$\tilde{w}(i, j) - p_j(c) \geq (1 + \varepsilon)^r$}
                        \State \Call{Match}{$i, \lrv{j, c}$}, $T \gets T \cup \{ \lrv{j, c} \}$
                    \EndIf
                \EndIf
            \EndWhile
            \For{$\lrv{j, c} \in T$}
                \State $\gamma_{i, c} \gets \tilde{w}(i, j) - p_j(c) - (1-\varepsilon)(1+\varepsilon)^{r_i+1}$
                \Comment{Bid value}
                \State $p_j(c) \gets p_j(c) + \gamma_{i, c}$ 
                \label{line:price_update}
            \EndFor
            \State $I \gets I \sm \lrc{i}$ 
        \EndProcedure

        \Procedure{Match}{$i, \lrv{j, c}$}
            \If{$c$ is already assigned to another bidder $y \neq i$} 
                \State $F \gets F \sm \lrc{\lrp{y, j}}$  
                \If{$Q_y \neq \emptyset$}
                    $I \gets I \cup \lrc{y}$
                \EndIf
            \EndIf
            \State $F \gets F \cup \{ (i, j) \}$
        \EndProcedure
    \end{algorithmic}
    \end{minipage}
\end{algorithm*}

\subsection{Invariants and Analysis}
Throughout the runtime of the algorithm, we maintain the following invariants. 
\begin{Invariant}
    Fix any $i \in A$. 
    For all $k \in N(i) \sm F(i)$ and $l \in M(k)$, $\tilde{w}(i, k) - p_k(l) < \max \{(1 + \varepsilon)^{r_i + 1}, (1 + \varepsilon)^{r_{ik} - s_{\min}}\}$. 
    \label{invariant:mult_util_ub}
\end{Invariant}
\begin{proof}
    This is true at the start since $F = \emptyset$, $r_i = \max_{j \in N(i)} \{ \textsc{iLog}(\tilde{w}(i, j)) \}$, and all object copy 
    prices are set to zero so 
    $\tilde{w}(i, j) - p_j(c) = \tilde{w}(i, j) < (1+\varepsilon)^{r_i + 1}$
    for all $j \in N(i)$ and $c \in M(j)$. We show in \cref{invariant:monotonic_price_increase} that throughout the algorithm the prices 
    of object copies are monotonically increasing. Thus, it suffices to show the inequality holds whenever $r_i$ changes. 
    Note that $r_i$ can only ever decrease. If $r_i$ changes to some value $r$, then we can guarantee there exists no pairs 
    $(r^\prime, (i, k))$ where $r^\prime \geq r+1 > r = r_i$ and $k \in N(i) \sm F(i)$ in $Q_i$ as such pairs must have been popped and discarded. 
    Additionally, we can guarantee that a lower bound on the indices of any pair popped for a specific $k$ is $r_{ik}-s_{\min}$. Thus, there exists 
    no $k \in N(i) \sm F(i)$ and $l \in M(k)$ with $\tilde{w}(i, k) - p_k(l) \geq \{(1 + \varepsilon)^{r_i + 1}, (1 + \varepsilon)^{r_{ik} - s_{\min}}\}$ by construction
    and hence the claim follows. 
    \qed
\end{proof}

\begin{Invariant}
    The prices of object copies are monotonically increasing. 
    \label{invariant:monotonic_price_increase}
\end{Invariant}
\begin{proof}
    Fix some bidder $i$. Suppose $(r, (i, j))$ was popped from $Q_i$ and $i$ chooses to add
    some tuple $\lrv{j, c}$ to $T$. By construction, $\tilde{w}(i, j) - p_j(c) \geq (1+\varepsilon)^r$. 
    Thus, 
    \begin{align*}
        \gamma_{i, c} = \tilde{w}(i, j) - p_j(c) - (1-\varepsilon)(1+\varepsilon)^{r_i+1} \geq (1+\varepsilon)^r - (1-\varepsilon^2)(1+\varepsilon)^{r_i} > 0
    \end{align*}
    where the last inequality follows from the fact that $r_i \leq r$.
    \qed
\end{proof}

\begin{Invariant}
    At the end of a call to $\Call{AssignAndBid}{i}$, $i$ is strongly $\varepsilon$-happy. 
    \label{invariant:mult_strong_eps_happ} 
\end{Invariant}
\begin{proof}
    Consider a tuple $\lrv{j, c} \in T$ during this iteration.  
    Right before we updated $p_j(c)$, \cref{invariant:mult_util_ub} implies that for all $k \in N(i) \sm F(i)$ and $l 
    \in M(k)$, $\tilde{w}(i, k) - p_k(l) < \max\{(1 + \varepsilon)^{r_i + 1}, (1 + \varepsilon)^{r_{ik} - s_{\min}}\}$. 
    Partition $N(i) \sm F(i)$ into the set of objects $L_1$ that are less than the first argument 
    of the max and the set $L_2$ that are less than the second argument. 
    For $k \in L_2$, we have that for all $l \in M(k)$
    \begin{equation*}
        \tilde{w}(i, k) - p_{k}(l) < (1+\varepsilon)^{r_{ik}-s_{\min}} \leq \varepsilon (1+\varepsilon)^{r_{ik}} = \varepsilon \tilde{w}(i, k),  
    \end{equation*}
    which implies $(1-\varepsilon)\tilde{w}(i, k) - p_{k}(l) < 0$.
    For $k \in L_1$, we have that for all $l \in M(k)$, 
    \begin{equation*}
        (1-\varepsilon)\tilde{w}(i, k) - p_{k}(l) < (1-\varepsilon)(\tilde{w}(i, k) - p_{k}(l)) < (1-\varepsilon)(1+\varepsilon)^{r_i+1}. 
    \end{equation*}
    By construction of the bid, we have that $\tilde{w}(i, j) - p_j(c) = (1-\varepsilon)(1+\varepsilon)^{r_i+1}$. Thus
    \begin{equation*}
        \tilde{w}(i, j) - p_j(c) \geq \max \lrc{ \max_{k \in N(i) \sm F(i), \; l \in M(k)} \{(1-\varepsilon)\tilde{w}(i, k) - p_k(l)\}, 0}
    \end{equation*}
    and the edge $(i, j)$ is strongly $\varepsilon$-happy. 

    Now consider each edge $(i, g)$, where $i$ is assigned to $h \in M(g)$, that was already included in $F$ at the start of the 
    call but $\lrv{g, h} \notin T$ at the end of the call. By the previous analysis, we know that when $i$ last increased $p_g(h)$, 
    the edge $(i, g)$ must have been strongly $\varepsilon$-happy. By \cref{invariant:monotonic_price_increase} we know that the prices 
    of object copies are monotonically increasing throughout the algorithm, so the edge $(i, g)$ must still be strongly $\varepsilon$-happy. 
    
    If at the end of the call $i$ is saturated, then $i$ must be strongly $\varepsilon$-happy. However, if $i$ is unsaturated 
    then it must be that $Q_i$ is empty. In this case for each $k \in N(i) \sm F(i)$ it must be that 
    the entry $(r_{ik}-s_{\min}, (i, k))$ must have been popped 
    from $Q_i$, which indicates that for all $l \in M(k)$,
    \begin{equation*}
        \tilde{w}(i, k) - p_k(l) < (1+\varepsilon)^{r_{ij}-s_{\min}} \leq \varepsilon (1+\varepsilon)^{r_{ij}} = \varepsilon \tilde{w}(i, k).
    \end{equation*}
    This gives $(1-\varepsilon)\tilde{w}(i, k) - p_k(l) < 0$, so
    $\pi(i) = 0$ and $i$ is still strongly $\varepsilon$-happy. 
    \qed
\end{proof}

By \cref{invariant:mult_strong_eps_happ}, when \cref{alg:mult_auction_bmatching} terminates, the 
\bmatching{} $F$ returned must be strongly $\varepsilon$-happy. Termination is guaranteed since for each bidder $i\in A$, $Q_i$ is 
a fixed length and in case it is emptied, $i$ must be strongly $\varepsilon$-happy.  
We can also guarantee that if an object $j \in B$ is unsaturated at termination, then $\min_{c \in M(j)} \lrc{p_j(c)} = 0$. 
If we set $p(j) = \min_{c \in M(j)} \lrc{p_j(c)}$, using \cref{prop:strong_e_hap_to_e_hap_mult}
and \cref{lem:mult_loss_bmatching}, $F$ must be a $(1-\varepsilon)$-approximate with respect to the weights
$\tilde{w}$. By \cref{prop:scaling}, $F$ is then $(1-\varepsilon^\prime)$-approximate with respect to
$w$. 

For runtime analysis, we divide the algorithm into two phases: initialization and auction. The initialization phase populates $Q_i$ for each $i \in A$, which requires bucket sorting $ms_{\min}$ pairs and
takes $\bigO{ms_{\min} + (s_{\max}+s_{\min})} = \bigO{m\varepsilon^{-1} \log \varepsilon^{-1}}$ time. The  auction phase 
involves calls to $\Call{AssignAndBid}{i}$ for each $i \in A$, which is dominated by the size of $Q_i$ and the time it takes 
to find a minimum price object copy and update its price. If we use a min priority queue to maintain an ordering of 
$M(j)$ for each $j \in B$, where the price $p_j(c)$ of an object copy $c \in M(j)$ is its priority, 
then these operations takes $\bigO{1}$ and $\bigO{\log \beta}$ time, respectively.
We may need to do a price update for each pair in $Q_i$, so the total amount of work done in all calls to 
$\Call{AssignAndBid}{i}$ is $\bigO{\degree{i}s_{\min}\log\beta}$. Summing over all bidders $i \in A$, the total amount of 
work done in the auction phase is $\bigO{m s_{\min}\log\beta} = \bigO{m\varepsilon^{-1} \log \varepsilon^{-1} \log \beta}$. 
Since the weight preprocessing described in \cref{subsec:weight} takes $\bigO{m}$ time, we have the following. 

\begin{theorem}
    There exists a multiplicative auction algorithm for \maxwbm{} that gives a $(1-\varepsilon)$-approximate solution in $\bigO{m \varepsilon^{-1} \log \varepsilon^{-1} \log \beta}$ time. 
\end{theorem}
\section{Conclusions} \label{sec:conclusions}

In this paper, 
we present a near-linear $(1-\varepsilon)$-approximate multiplicative auction algorithm for \mwbm{}. However, the algorithm's runtime has a dependence on $\beta$, the max $b$-value. While it is 
reasonable to assume $\beta = \bigO{1}$ in many practical applications, 
a direction for further research would be to remove this dependence on 
$\beta$ to have runtime parity with the algorithm of Huang and Pettie \cite{huang2022approximate}. 

\bibliographystyle{splncs04}
\bibliography{refs.bib} 

\end{document}